%% file: main.tex
\documentclass[copyright, creativecommons]{eptcs}

\providecommand{\event}{FICS 2013}  
\usepackage{breakurl}             
\usepackage{amssymb} 
\usepackage{amsmath} 
\usepackage{microtype}
\usepackage{url}
\usepackage{proof}
\newcommand{\sol}[2]{{\mathcal S}(\seq{#1}{#2})}
\newcommand{\Sol}[2]{{\mathcal S}\bigl(\seq{#1}{#2}\bigr)} 

\newcommand{\vect}[1]{\overrightarrow{#1}} 

\newcommand{\impl}{\supset}
\newcommand{\cosign}{\textit{co}}
\newcommand{\cool}{\overline{\lambda}^\cosign}
\newcommand{\fix}{{\nu}}  
\newcommand{\gfp}{\mathsf{gfp}}  
\newcommand{\oo}{\mathbb{O}}
\newcommand{\lb}{\lambda}
\newcommand{\LAx}{\textit{LAx}}
\newcommand{\RIntro}{\textit{RIntro}}
\newcommand{\LIntro}{\textit{LIntro}}
\newcommand{\LVecIntro}{\textit{LVecIntro}}
\newcommand{\Alts}{\textit{Alts}}
\newcommand{\App}{\textit{App}}
\newcommand{\ol}{\overline{\lambda}}
\newcommand{\olH}{\overline{\lambda}_{\sf Horn}}
\newcommand{\coolfs}{\overline{\lambda}^\cosign_{\Sigma}}
\newcommand{\coolfsE}{E\overline{\lambda}^\cosign_{\Sigma}}
\newcommand{\olfsfix}{\overline{\lambda}^\gfp_{\Sigma}}
\newcommand{\seqt}[3]{#1\vdash #2:#3}
\newcommand{\seql}[4]{#1|#2:#3\vdash #4}
\newcommand{\seq}[2]{#1\Rightarrow #2}
\newcommand{\tuple}[1]{\langle #1 \rangle}

\newcommand{\der}[2]{#1#2}
\newcommand{\fl}[2]{\langle #2\rangle_{#1}}

\newcommand{\bs}[2]{#1+#2} 
\newcommand{\ns}[2]{#1+\cdots+#2} 
\newcommand{\s}[2]{\sum\limits_{#1}^{}{#2}} 

\newcommand{\fs}[2]{\s{#1}{#2}}

\newcommand{\interp}[2]{[\![#1]\!]_{#2}}
\newcommand{\interpwe}[1]{[\![#1]\!]}

\newcommand{\N}[2]{N_{#1}(#2)}

\newcommand{\colbase}{{\sf mem}}
\newcommand{\colr}[2]{\colbase(#1,#2)}
\newcommand{\colebase}{\colbase_E}
\newcommand{\colra}[2]{\colebase(#1,#2)}
\newtheorem{lemma}{Lemma}  
\newtheorem{theorem}[lemma]{Theorem}  

\newtheorem{example}[lemma]{Example} 
\newtheorem{definition}[lemma]{Definition} 
\newtheorem{proposition}[lemma]{Proposition}
\newtheorem{corollary}[lemma]{Corollary}

\title{A Coinductive Approach to Proof Search}
\def\titlerunning{A Coinductive Approach to Proof Search} 
\author{Jos\'{e} Esp\'{\i}rito Santo
\institute{Centro de Matem\'{a}tica}
\institute{Universidade do Minho}
\institute{Portugal}
\and
Ralph Matthes
\institute{Institut de Recherche en Informatique de Toulouse (IRIT)}
\institute{C.N.R.S. and University of Toulouse}
\institute{France}
\and
Lu\'{\i}s Pinto
\institute{Centro de Matem\'{a}tica}
\institute{Universidade do Minho}
\institute{Portugal}
}

\def\authorrunning{J.~Esp\'{\i}rito Santo and  R.~Matthes and L.~Pinto}

\begin{document}

\maketitle
\begin{abstract}

  We propose to study proof search from a coinductive point of view. In this paper, we consider intuitionistic logic and a focused system
  based on Herbelin's LJT for the implicational fragment. We introduce a variant of lambda calculus with potentially
  infinitely deep terms and a means of expressing alternatives for the
  description of the ``solution spaces'' (called B\"ohm forests),
  which are a representation of all (not necessarily well-founded but
  still locally well-formed) proofs of a given formula (more
  generally: of a given sequent).

As main result we obtain, for each given formula, the reduction of a coinductive definition of the solution space to a effective coinductive description 
in a finitary term calculus with a formal greatest fixed-point operator. This reduction works in a quite direct manner for the case of Horn formulas. For  the general case, the naive extension would not even be true. We  need to study ``co-contraction'' of contexts (contraction bottom-up) for dealing with the varying contexts needed beyond the Horn fragment, and we point out the appropriate finitary calculus, where fixed-point variables are typed with sequents. Co-contraction enters the interpretation of the formal greatest fixed points - curiously in the semantic interpretation of fixed-point variables and not of the fixed-point operator.

\end{abstract}

\input{intro}
\input{background}
\input{coinductive-representation}
\input{finitary-representation}
\input{equivalence}
\input{conclusion}

\end{document}

%% file: intro.tex
\section{Introduction}\label{sec:intro}

Proof theory starts with the observation that a proof is more than
just the truth value of a theorem. A valid theorem can have many
proofs, and several of them can be interesting. In this paper, we
somehow extend this to the limit and study all proofs of a given
proposition. Of course, who studies proofs can also study any of them
(or count them, if there are only finitely many possible proofs, or
try to enumerate them in the countable case). But we do this study somehow
simultaneously: we introduce a language to express the full ``solution
space'' of proof search. And since we focus on the generative aspects
of proof search, it would seem awkward to filter out failed proof
attempts from the outset. This does not mean that we pursue impossible
paths in the proof search (which would hardly make sense) but that we
allow to follow infinite paths. An infinite path does not correspond
to a successful proof, but it is a structure of locally correct proof
steps. In other words, we use coinductive syntax to model \emph{all}
locally correct proof figures. This gives rise to a not necessarily
wellfounded search tree. However, to keep the technical effort
simpler, we have chosen a logic where this tree is finitely branching,
namely the implicational fragment of intuitionistic propositional
logic (with proof system given by the cut-free fragment of the system
$\ol$ by Herbelin \cite{HerbelinCSL94}).

Lambda terms or variants of them (expressions that may have bound
variables) are a natural means to express proofs (an observation that is called \emph{the}
Curry-Howard isomorphism) in implicational logic. Proof alternatives
(locally, there are only finitely many of them since our logic has no
quantifier that ranges over infinitely many individuals) can be
formally represented by a finite sum of such solution space
expressions, and it is natural to consider those sums up to
equivalence of the \emph{set} of the alternatives. Since infinite
lambda-terms are involved and since whole solution spaces are being
modeled, we call these coinductive terms \emph{B\"ohm forests}.

By their coinductive nature, B\"ohm forests are no proper syntactic
objects: they can be defined by all mathematical (meta-theoretic)
means and are thus not ``concrete'', as would be expected from
syntactic elements. This freedom of definition will be demonstrated
and exploited in the canonical definition (Definition~\ref{def:sol})
of B\"ohm forests as solutions to the task of proving a sequent (a
formula $A$ in a given context $\Gamma$). In a certain sense, nothing
is gained by this representation: although one can calculate on a
case-by-case basis the B\"ohm forest for a formula of interest and see
that it is described as fixed point of a system of equations
(involving auxiliary B\"ohm forests as solutions for the other
meta-variables that appear in those equations), an arbitrary B\"ohm
forest can only be observed to any finite depth, without ever knowing
whether it is the expansion of a regular cyclic graph structure (the latter being a finite structure).

Our main result is that the B\"ohm forests that appear as solution
spaces of sequents have such a finitary nature: more precisely, they
can be interpreted as semantics of a finite term in a variant of
lambda calculus with alternatives and formal greatest
fixed-points. For the Horn fragment (where nesting of implications to
the left is disallowed), this works very smoothly without surprises
(Theorem~\ref{thm:Horn}). The full implicational case, however, needs
some subtleties concerning the fixed-point variables over which the
greatest fixed points are formed and about capturing redundancy that
comes from the introduction of several hypotheses that suppose the
same formula. The interpretation of the finite expressions in terms of
B\"ohm forests needs a special operation that we call
\emph{co-contraction} (contraction bottom-up). However, this operation is already definable in terms
of B\"ohm forests. Without this operation, certain repetitive patterns
in the solution spaces due to the presence of negative occurrences of
implications could not be identified. With it, we obtain the finitary
representation (Theorem~\ref{thm:FullProp}).

\medskip In the next section, we quickly recapitulate syntax and
typing rules of the cut-free fragment of system $\ol$ and also
carefully describe its restriction to Horn formulas.

Section~\ref{sec:coindrepr} has the definition of the not necessarily
well-founded proofs, corresponding to a coinductive reading of $\ol$
(including its typing system). This is system $\cool$. Elimination
alternatives are then added to this system (yielding the B\"ohm
forests), which directly allow the definition of the solution spaces
for the proof search for sequents. We give several examples and then
show that the defined solution spaces adequately represent all the
$\cool$ proofs of a sequent.

In Section~\ref{sec:finitaryrepr}, we present first the finitary
system to capture the Horn fragment and then modify it to get the main
result for full implicational logic.

The paper closes with discussions on related and future work in Section~\ref{sec:conclusion}.


%% file: background.tex
\section{Background}\label{sec:background}

We recall below the cut-free fragment of system $\ol$
(a.k.a. LJT), a sequent calculus for intuitionistic implication
by Herbelin \cite{HerbelinCSL94}.

Letters $p,q,r$ are used to range over a base set of propositional
variables (which we also call \emph{atoms}).  Letters $A,B,C$ are used
to range over the set of formulas (= types) built from propositional
variables using the implication connective (that we write $A\impl B$)
that is parenthesized to the right. Often we will use the fact that
any implicational formula can be uniquely decomposed as $A_1\impl A_2\impl
\ldots\impl A_n\impl p$ with $n\geq0$, also written in vectorial
notation as $\vec{A}\impl p$. For example, if the vector $\vec{A}$ is
empty the notation means simply $p$, and if $\vec{A}=A_1,A_2$, the
notation means $A_1\impl(A_2\impl p)$.

The cut-free expressions of $\ol$ are separated into terms and lists,
and are given by:
$$
\begin{array}{lcrcl}
\textrm{(terms)} &  & t,u & ::= & \der xl\,|\,\lambda x^A.t\\
\textrm{(lists)} &  & l & ::= & \tuple{}\,|\,u::l\\
\end{array}
$$
where a countably infinite set of variables ranged over by letters
$x$, $y$, $w$, $z$ is assumed. Note that in lambda-abstractions we
adopt a \emph{domain-full} presentation, annotating the bound variable
with a formula.  The term
constructor $\der{x}{l}$ is usually called \emph{application}. Usually
in the meta-level we prefer to write $x \tuple{t_1,\ldots,t_n}$ (with
$n\in\mathbb{N}_0$) to range over application constructions, and avoid
speaking about lists explicitly (where obviously, the notation
$\tuple{t_1,\ldots,t_n}$ means $\tuple{}$ if $n=0$ and $t_1::l$, if
$\tuple{t_2,\ldots,t_n}$ means $l$). In the meta-level, when we know $n=0$, instead of $x \tuple{t_1,\ldots,t_n}$, we simply write the variable $x$. 

We will view contexts $\Gamma$
as finite lists of
declarations $x:A$, where no variable $x$ occurs twice. The
context $\Gamma,x:A$ is obtained from $\Gamma$ by adding the
declaration $x:A$, and will only be written if this yields again a
valid context, i.\,e., if $x$ is not declared in $\Gamma$.
The system has a form of sequent for each class of expressions:
$$\Gamma\vdash t:A\qquad \quad \Gamma|l:A\vdash p.$$
Note the restriction to \emph{atomic sequents} (the RHS formula is an atom) in the case of list sequents.

The rules of $\ol$ for deriving sequents are in Figure~\ref{fig:lambda-bar}.
\begin{figure}[tb]\caption{Typing rules of $\ol$}\label{fig:lambda-bar}
$$
\begin{array}{c}
\infer[\LAx]{\seql{\Gamma}{\tuple{}}{p}{p}}{}\quad\quad\infer[\LIntro]{\seql{\Gamma}{u::l}{A\impl B}{ p}}{\Gamma\vdash u:A&
\seql{\Gamma}{l}{B}{p}}\\\\
\infer[\RIntro]{\Gamma\vdash\lambda x^A.t:A\impl B}{\Gamma,x:A\vdash
t:B}\quad\quad
\infer[\App]
 {\seqt{\Gamma}{\der{y}{l}}{p}}
 {\seql{\Gamma}{l}{A}{p}\quad\quad (y:A)\in\Gamma}
\end{array}
$$
\end{figure}
Note that, as list sequents are atomic, the conclusion of the
application rule is also atomic. This is not the case in Herbelin's
original system \cite{HerbelinCSL94}, where list sequents can have a
non-atomic formula on the RHS. In the variant of cut-free $\ol$ we
adopted, the only rule available for deriving a term sequent whose RHS
is an implication is $\RIntro$. Still, our atomic restriction will not
cause loss of completeness of the system for intuitionistic
implication. This restriction is typically adopted in systems tailored
for proof search, as for example systems of focused proofs. In fact,
$\ol$ corresponds to a focused backward chaining system where all
atoms are \emph{asynchronous} (see e.\,g. Liang and Miller
\cite{LiangMillerTCS09}).

We will need the following properties of $\ol$.

\begin{lemma}[Type uniqueness]\label{lem:type-unique-lambda-bar}
  \begin{enumerate}
  \item Given $\Gamma$ and $t$, there is at most one $A$ such that $\seqt\Gamma tA$.
  \item Given  $\Gamma$, $l$ and $A$, there is at most one $p$ such that $\seql\Gamma lAp$.
  \end{enumerate}
\end{lemma}
\begin{proof}
  Simultaneous induction on derivability.
\end{proof}
Since the empty list $\tuple{}$ has no type index, we need to know $A$ in the second statement of the previous lemma.

\begin{lemma}[Inversion of typing]\label{lem:gen-lambda-bar} In $\ol$:
\begin{enumerate}
\item $\seqt\Gamma{\lambda x^A.t}B$ iff there exists $C$ s.t. $B=A\impl C$ and $\seqt{\Gamma, x:A}tC$;
\item $\seqt{\Gamma}{x \tuple{t_1,\ldots,t_k}}A$ iff $A=p$ and there exists $\vec B$ s.t. $x:\vec B\impl p\in\Gamma$ and  $\seqt{\Gamma}{t_i}{B_i}$, for any $i$.
\end{enumerate}
\end{lemma}
\begin{proof}
1. is immediate and 2. follows with the help of the fact that:  $\seql{\Gamma}{\tuple{t_1,\ldots,t_k}}Bp$  iff there exist $B_1,...,B_k$ s.t. $B=B_1\impl...\impl B_k\impl p$ and, for any $i$, $\seqt{\Gamma}{t_i}{B_i}$  (proved by induction on $k$).
\end{proof}

Now we identify the \emph{Horn fragment} of cut-free $\ol$, that we
denote by $\olH$. The class of \emph{Horn formulas} (also called
\emph{Horn clauses}) is given by the grammar:
$$
\begin{array}{lcrcl}
\textrm{(Horn formulas)} &  & H & ::= & p\,|\,p\impl H\\
\end{array}
$$
where $p$ ranges over the set of propositional variables. Note that
for Horn formulas, in the vectorial notation $\vec{H}\impl p$, the
vector components $H_i$ are necessarily propositional variables, i.\,e.,
any Horn formula is of the form $\vec q\impl p$.

The Horn fragment is obtained by restricting sequents as follows:
\begin{enumerate}
\item contexts are restricted to \emph{Horn contexts}, i.\,e., contexts where all formulas are Horn formulas;
\item term sequents are restricted to atomic sequents, i.\,e., term
  sequents are of the form $\seqt{\Gamma}tp$.
\end{enumerate}
As a consequence, the $\lambda$-abstraction construction and the rule
$RIntro$, that types it, are no longer needed. The restricted typing
rules are presented in Figure~\ref{fig:lambda-bar-Horn}. 
\begin{figure}[tb]\caption{Typing rules of $\olH$}\label{fig:lambda-bar-Horn}
$$
\begin{array}{c}
\infer[\LAx]{\seql{\Gamma}{\tuple{}}{p}{p}}{}\quad\quad\infer[\LIntro]{\seql{\Gamma}{u::l}{p\impl H}{ q}}{\Gamma\vdash u:p&
\seql{\Gamma}{l}{H}{q}}\\\\
\infer[\App]
 {\seqt{\Gamma}{\der{y}{l}}{p}}
 {\seql{\Gamma}{l}{H}{p}\qquad\qquad(y:H)\in\Gamma}
\end{array}
$$
\end{figure}


%% file: coinductive-representation.tex
\section{Coinductive representation of proof search in lambda-bar}\label{sec:coindrepr}

We want to represent the whole search space for cut-free proofs in
$\ol$. This is profitably done with coinductive structures. Of course,
we only consider locally correct proofs. Since proof search may fail
when infinite branches occur (depth-first search could be trapped
there), we will consider such infinite proofs as proofs in an extended
sense and represent them as well, thus we will introduce expressions
that comprise all the possible well-founded and non-wellfounded proofs
in cut-free $\ol$.

The raw syntax of these possibly non-wellfounded proofs is presented as follows
$$ N ::=_\cosign \lambda x^A.N\,|\,  x \tuple{N_1,\ldots,N_k}\enspace,$$
yielding the (co)terms of system $\cool$ (read coinductively, as indicated by the index $\cosign$). Note that instead of a formal class of lists $l$ as in the $\ol$-system, we adopt here the more intuitive notation $\tuple{N_1,\ldots,N_k}$  to represent finite lists.

Since the raw syntax is interpreted coinductively, also the typing
rules have to be interpreted coinductively, which is symbolized by the
double horizontal line in Figure~\ref{fig:co-lambda-bar}, a notation
that we learnt from Nakata, Uustalu and Bezem~\cite{NUB}. (Of course, the formulas/types stay
inductive.)
As expected, the restriction of the typing relation to the finite $\ol$-terms coincides with the typing relation of the $\ol$ system:

\begin{figure}[tb]\caption{Typing rules of $\cool$}\label{fig:co-lambda-bar}
$$
\begin{array}{c}
\infer=[\RIntro]{\seqt\Gamma{\lambda x^A.t}{A\impl B}}{\seqt{\Gamma,x:A}
tB}\quad\quad
\infer=[\LVecIntro]
 {\seqt{\Gamma}{x\tuple{N_1,\ldots,N_k}}{p}}
 {(x:B_1,\ldots,B_k\impl p)\in\Gamma\quad\seqt\Gamma{N_i}{B_i}, i=1,\ldots, k}
\end{array}
$$
\end{figure}

\begin{lemma}
\label{lem:equiv-typability-lambda-bar-terms}
For any $t\in\ol$,  $\seqt\Gamma t A$ in $\ol$ iff $\seqt\Gamma t A$ in $\cool$.
\end{lemma}
\begin{proof}
By induction on $t$, with the help of  Lemma \ref{lem:gen-lambda-bar}.
\end{proof}

\begin{example}
\label{ex:Churchinfty}
Consider $\omega:=\lambda f^{p\impl
  p}.\lambda x^p.N$ with $N = f\tuple N$ of type $p$. This infinite
term $N$ is also denoted $f^\infty$.
\end{example}

It is quite common to describe elements of coinductive syntax by
(systems of) fixed point equations. As a notation on the
\emph{meta-level} for unique solutions of fixed-point equations, we
will use the binder $\fix$ for the solution, writing $\fix\, N. M$,
where $N$ typically occurs in the term $M$. Intuitively, $\fix\, N. M$
is the $N$ s.\,t. $N=M$. (The letter $\nu$ indicates interpretation in
coinductive syntax.)

\begin{example}
\label{ex:Churchinftycont}
$\omega$ of Example~\ref{ex:Churchinfty} can be written as
$\lambda f^{p\impl p}.\lambda x^p.\fix\, N.f\tuple N$.
$\seqt{\Gamma,f:p\impl p,x:p}{\fix\, N.f\tuple N}p$ is seen
coinductively, so we get $\seqt\Gamma\omega{(p\impl p)\impl p\impl p}$.
\end{example}

We now come to the representation of whole search spaces.  The set of
coinductive cut-free $\ol$-terms with finite numbers of elimination
alternatives is denoted by $\coolfs$ and is given by the following
grammar:

$$
\begin{array}{lcrcl}
\textrm{(co-terms)} &  & N & ::=_\cosign & \lambda x^A.N\,|\, \ns{E_1}{E_n}\\ 
\textrm{(elim. alternatives)} &  & E & ::=_\cosign & x \tuple{N_1,\ldots,N_k}\\
\end{array}
$$
where both $n,k\geq0$ are arbitrary. Note that summands cannot be lambda-abstractions.\footnote{The division into two syntactic categories also forbids the generation of an infinite sum (for which $n=2$ would suffice had the categories for $N$ and $E$ been amalgamated).} We will often use $\fs i{E_i}$
instead of $\ns{E_1}{E_n}$ 
if the dependency of $E_i$ on $i$
is clear, as well as the number of elements. Likewise, we write $\fl i
{N_i}$ instead of $\tuple{N_1,\ldots,N_k}$. If $n=0$, we write $\oo$
for $\ns{E_1}{E_n}$. 
If $n=1$, we write $E_1$ for $\ns{E_1}{E_n}$
(in particular
this injects the category of elimination alternatives into
the category of co-terms) and do as if $+$ was a binary operation on
(co)terms. However, this will always have a unique reading in terms of
our raw syntax of $\coolfs$. In particular, this reading makes $+$
associative and $\oo$ its neutral element.

Co-terms of $\coolfs$ will also be called B\"ohm forests. Their
coinductive typing rules are the ones of $\cool$, together with the rule given in
Figure~\ref{fig:co-lambda-bar-sum}, where the sequents for (co)terms
and elimination alternatives are not distinguished notationally.

\begin{figure}[tb]\caption{Extra typing rule of $\coolfs$ w.\,r.\,t.~$\cool$}\label{fig:co-lambda-bar-sum}
$$
\begin{array}{c}
\infer=[\Alts]{\seqt{\Gamma}{\ns{E_1}{E_n}}p}{\seqt\Gamma{E_i}p, i=1,\ldots, n}\\\\
\end{array}
$$
\end{figure}
Notice that $\seqt{\Gamma}{\oo}p$ for all $\Gamma$ and $p$.

Below we consider sequents $\seq{\Gamma}A$ with $\Gamma$ a context and
$A$ an implicational formula (corresponding to term sequents of $\ol$
without proof terms -- in fact, $\seq{\Gamma}A$ is nothing but the
pair consisting of $\Gamma$ and $A$, but which is viewed as a problem
description: to prove formula $A$ in context $\Gamma$).

\begin{definition}\label{def:sol}
  The function ${\mathcal S}$, which takes a sequent $\seq{\Gamma}A$
  and produces a B\"ohm forest which is a coinductive representation
  of the sequent's solution space, is given corecursively as follows:
  In the case of an implication,
$$\sol{\Gamma}{A\supset B} := \lambda x^A.\sol{\Gamma,x:A}{B}\enspace,$$
since $\RIntro$ is the only way to prove the implication.

In the case of an atom $p$, for the definition of $\sol{\Gamma}{p}$, let $y_i:A_i$ be the $i$-th variable in $\Gamma$ with $A_i$ of the form $\vec{B_i}\impl p$. Let $\vec{B_i} = B_{i,1},\ldots,B_{i,k_i}$. Define $N_{i,j}:=\sol{\Gamma}{B_{i,j}}$. Then, $E_i:=y_i\fl j{N_{i,j}}$, and finally,
$$\sol{\Gamma}{p} := \fs i{E_i}\enspace.$$
This is more sloppily written as
$$\sol{\Gamma}{p} :=  \fs{{y:\vec{B}\supset p\in\Gamma}} {y\fl{j}{\sol{\Gamma}{B_j}}}\enspace.$$
In this manner, we can even write the whole definition in one line:
$$\sol{\Gamma}{\vec A \impl p} :=  \lambda \vec x:\vec A.\fs{{y:\vec{B}\supset p\in\Delta}} {y\fl{j}{\sol{\Delta}{B_j}}}\quad\mbox{with $\Delta:=\Gamma,\vec x:\vec A$}$$
\end{definition}

This is a well-formed definition: for every $\Gamma$ and $A$, $\sol\Gamma A$ is a B\"ohm forest and as such rather a semantic object.

\begin{lemma}
  Given $\Gamma$ and $A$, the typing  $\seqt\Gamma{\sol\Gamma A}A$ holds in $\coolfs$.
\end{lemma}

Let us illustrate  the  function $\cal S$ at work with some examples.

\begin{example}
\label{ex:Church}
We consider first the formula $A= (p\impl p)\impl p\impl p$ and the empty context. We have:
$$
\sol{}{(p\impl p)\impl p\impl p}\\
=\lambda f^{p\impl p}.\lambda x^p. \sol{f:p\impl p,x:p}p
$$
Now, observe that $\sol{f:p\impl p,x:p}p=\bs{f\tuple{\sol{f:p\impl p,x:p}p}}x$. 
We
identify $\sol{f:p\impl p,x:p}p$ as the solution for $N$ of the equation $N=\bs{f\tuple N}x$. 
Using $\fix$ as means to communicate solutions of fixed-point equations on the \emph{meta-level} as for $\cool$, we have
$$
\sol{}{(p\impl p)\impl p\impl p}
=\lambda f^{p\impl p}.\lambda x^p. \fix\, N. \bs{f\tuple N}x 
$$

By unfolding of the fixpoint and by making a choice at each of the elimination alternatives, we can \emph{collect} from this coterm as the finitary solutions of the sequent all the Church numerals ($\lambda f^{p\impl p}.\lambda x^p.f^n\tuple{x}$ with $n\in\mathbb{N}_0$), together with the infinitary solution $\lambda f^{p\impl p}.\lambda x^p.f^\infty$,
studied before as example for $\cool$ (corresponding to always making the
$f$-choice at the elimination alternatives).
\end{example}

\begin{example}
\label{ex:Horn}
We consider now an example in the Horn fragment. Let $\Gamma=x:p\impl
q\impl p,y:q\impl p\impl q,z:p$ (again with $p\neq q$). Note that the
solution spaces of $p$ and $q$ relative to this sequent are mutually
dependent and they give rise to the following system of equations:
$$
\begin{array}{rcl}
N_p&=&\bs{x\tuple{N_p,N_q}}z\\ 
N_q&=&y\tuple{N_q,N_p}\\
\end{array}
$$
and so we have
$$
\begin{array}{rcl}
\sol\Gamma p&=& \fix\, N_p.\bs{x\tuple{N_p,\fix\, N_q.y\tuple{N_q,N_p}}}z\\
\sol\Gamma q&=& \fix\, N_q.y\tuple{N_q,\fix\, N_p. \bs{x\tuple{N_p,N_q}}z}\\ 
\end{array}
$$
Whereas for $p$ we can collect one finite solution ($z$), for $q$ we
can only collect infinite solutions. Because in the Horn case the
recursive calls of the $\cal S$ function are all relative to the same
(initial) context, in this fragment the solution space of a sequent
can always be expressed as a finite system of equations (one for each
atom occurring in the sequent), see Theorem~\ref{thm:Horn}.
\end{example}

\begin{example}
\label{ex:dn-Pierce}
Let us consider one further example where $A=((((p\impl q)\impl
p)\impl p)\impl q)\impl q$ (a formula that can be viewed as double
negation of Pierce's law, when $q$ is viewed as absurdity). We have
the following (where in sequents we omit formulas on the LHS)
$$
\begin{array}{rcl}
N_0&=&\sol{}{A}=\lambda x^{(((p\impl q)\impl p)\impl p)\impl q}. N_1\\
N_1&=&\sol{x}{q}= x\tuple{N_2}\\
N_2&=&\Sol{x}{((p\impl q)\impl p)\impl p}=\lambda y^{(p\impl q)\impl p}. N_3\\
N_3&=&\sol{x,y}{p}= y\tuple{N_4}\\
N_4&=&\sol{x,y}{p\impl q}=\lambda z^{p}. N_5\\
N_5&=&\sol{x,y,z}{q}= x\tuple{N_6}\\
N_6&=&\Sol{x,y,z}{((p\impl q)\impl p)\impl p}=\lambda y_1^{(p\impl q)\impl p}. N_7\\
N_7&=&\sol{x,y,z,y_1}{p}= \bs{\bs{y\tuple{N_8}}z}{y_1\tuple{N_8}}\\ 
N_8&=&\sol{x,y,z,y_1}{p\impl q}= \lambda z_1^{p}. N_9\\
N_9&=&\sol{x,y,z,y_1,z_1}{q}\\
\end{array}
$$
Now, in $N_9$ observe that $y,y_1$ both have type $(p\impl q)\impl p$
and $z,z_1$ both have type $p$, and we are back at $N_5$ but with the
duplicates $y_1$ of $y$ and $z_1$ of $z$. Later, we will call
 this duplication phenomenon
 \emph{co-contraction}, and we will give a finitary description of $N_0$ and, more generally, of all $\sol{\Gamma}{A}$, see Theorem~\ref{thm:FullProp}. Of course, by taking the middle alternative in $N_7$, we obtain a finite proof, showing that $A$ is provable in $\ol$.
\end{example}

We now define a membership semantics for co-terms and elimination alternatives of $\coolfs$ in terms of sets of (co)terms in $\cool$.

The \emph{membership relations} $\colr MN$ and $\colra ME$ are contained in $\cool\times\coolfs$ and $\cool\times\coolfsE$ respectively (where $\coolfsE$ stands for the set of elimination alternatives of $\coolfs$) and are given coinductively by the rules in Fig. \ref{fig:collect}.

\begin{figure}[tb]\caption{Membership relations}\label{fig:collect}
$$
\begin{array}{c}
\infer=[]{\colr{\lambda x^A.M}{\lambda x^A.N}}{\colr{M}{N}}\quad\quad
\infer=[\mbox{(for some $i$)}]{\colr{M}{\ns{E_1}{E_n}}}{\colra{M}{E_i}}\\\\
\infer=[]{\colra{x\tuple{M_1,\ldots,M_k}}{x\tuple{N_1,\ldots,N_k}}}{\colr{M_1}{N_1}&\ldots&\colr{M_k}{N_k}}
\end{array}
$$
\end{figure}

\begin{proposition}\label{prop:adequacy-general-case}
For any $N\in\cool$, $\colr N{\sol\Gamma A}$ iff $\seqt\Gamma N A$ in $\cool$.
\end{proposition}
\begin{proof}
``If''. Consider the relations
$$
\begin{array}{l}
R:=\{(N,{\sol\Gamma A})\mid \seqt\Gamma N A\}\\
R_E:=\{(x\fl i{N_i},x\fl i{\sol\Gamma{B_i}})\mid (x:B_1,\ldots,B_k\impl p)\in\Gamma\wedge\seqt\Gamma {x\tuple{N_1,\ldots,N_k}} p\}
\end{array}
$$
It suffices to show that $R\subseteq\colbase$, but this cannot be
proven alone since $\colbase$ and $\colebase$ are defined
simultaneously. We also prove $R_E\subseteq\colebase$, and to prove both by
coinduction on the membership relations, it suffices to show that the
relations $R$, $R_E$ are \emph{backwards closed}, i.\,e.:
\begin{enumerate}
\item $({\lambda x^A.M},{\lambda x^A.N})\in R$ implies $(M,N)\in R$;
\item $({M},{\ns{E_1}{E_n}})\in R$ implies for some $i$, $(M,E_i)\in R_E$;
\item $({x\tuple{M_1,\ldots,M_k}},{x\tuple{N_1,\ldots,N_k}})\in R_E$ implies for all $i$, $({{M_i}},{{N_i}})\in R$
\end{enumerate}

We illustrate one case. Consider $(N,{\sol\Gamma A})\in R$, with
${\sol\Gamma A}=\ns{E_1}{E_n}$. We must show that, for some $i$,
$(N,E_i)\in R_E$. From ${\sol\Gamma A}=\ns{E_1}{E_n}$, we must have
$A=p$. Now, from $\seqt\Gamma N p$, there must exist
$(x:B_1,\ldots,B_k\impl p)\in\Gamma$ and $N_1,...,N_k$ s.\,t.
$N=x\tuple{N_1,\ldots,N_k}$. By definition of $\sol\Gamma A$, there is
$i$ s.\,t. $E_i=x\tuple{\sol\Gamma{B_1},\ldots,\sol\Gamma{B_k}}$.

``Only if''. By coinduction on the typing relation of $\cool$. This is
conceptually easier than the other direction since $\vdash$ is a
single coinductively defined notion. We define a relation $R$ for
which it is sufficient to prove $R\subseteq\vdash$:
$$
\begin{array}{l}
R:=\{(\Gamma,N,A)\mid \colr N{\sol\Gamma A}\}
\end{array}
$$
Proving $R\subseteq\vdash$ by coinduction amounts to showing that $R$ is backwards closed -- with respect to the typing relation of $\cool$, i.\,e., we have to show:
\begin{enumerate}
\item $(\Gamma,\lambda x^A.t,A\impl B)\in R$ implies $((\Gamma,x:A), t, B)\in R$;
\item $(\Gamma,x\tuple{N_1,\ldots,N_k},p)\in R$ implies the existence of $B_1,\ldots,B_k$ s.\,t. $(x:B_1,\ldots,B_k\impl p)\in\Gamma$ and, for all $i=1,\ldots,k$, $(\Gamma,N_i,B_i)\in R$.
\end{enumerate}
We show the second case (relative to rule $\LVecIntro$). So, we
have $\colr N{\sol\Gamma A}$ with $N=x\tuple{N_1,\ldots,N_k}$ and
$A=p$, and we need to show that, for some $(x:B_1,\ldots,B_k\impl
p)\in\Gamma$, we have, for all $i$, $\colr{N_i}{\sol\Gamma{B_i}}$.
Since $A=p$, ${\sol\Gamma A}=\ns{E_1}{E_n}$. Hence, the second rule for $\colbase$
was used to infer $\colr N{\sol\Gamma A}$, i.\,e., there is a $j$ s.\,t.
$\colra N{E_j}$. Therefore, $E_j=x\tuple{M_1,\ldots,M_k}$ with terms $M_1$,
\ldots, $M_k$, and, for all $i$, $\colr {N_i}{M_i}$. By the definition
of $\sol\Gamma A$, this means that there are formulas $B_1$, \ldots,
$B_k$ s.\,t. $(x:B_1,\ldots,B_k\impl p)\in\Gamma$ and, for all $i$,
$M_i=\sol\Gamma{B_i}$.
\end{proof}

\begin{example}
\label{ex:Pierce}
Let us  consider the case of Pierce's law that is not valid intuitionistically. We have (for $p\neq q$):
$$
\sol{}{((p\impl q)\impl p)\impl p}
=\lambda x^{(p\impl q)\impl p}. x\tuple{\lambda y^p.\oo}
$$
The fact that we arrived at $\oo$ and found no elimination alternatives on the way \emph{annihilates} the co-term
and implies there are no
terms in the solution space of $\seq{}{((p\impl q)\impl p)\impl p}$ (hence no
proofs, not even infinite ones). 
\end{example}

\begin{corollary}[Adequacy of the co-inductive representation of proof search in $\ol$]
For any $t\in\ol$, we have $\colr t{\sol\Gamma A}$ iff $\seqt\Gamma t A$ (where the latter is the
inductive typing relation of $\ol$).
\end{corollary}
\begin{proof}
By the proposition  above and Lemma \ref{lem:equiv-typability-lambda-bar-terms}.
\end{proof}


%% file: finitary-representation.tex
\section{Finitary representation of proof search in lambda-bar}\label{sec:finitaryrepr}

In the first section we define a calculus of finitary
representations. In the third section we obtain our main result
(Theorem \ref{thm:FullProp}): given $\seq\Gamma C$, there is a
finitary representation of $\sol\Gamma C$ in the finitary calculus.
To make the proof easier to understand, we first develop in the second
section the particular case of the Horn fragment.

\subsection{The finitary calculus}\label{subsec:finitary-calculus}
The set of inductive cut-free $\ol$-terms with finite numbers of
elimination alternatives, and a fixpoint operator is denoted by
$\olfsfix$ and is given by the following grammar (read inductively):

$$
\begin{array}{lcrcl}
\textrm{(terms)} &  & N & ::= & \lambda x^A.N\,|\, \gfp\,{X}.\ns{E_1}{E_n}\,|\,X\\
\textrm{(elim. alternatives)} &  & E & ::= & x \tuple{N_1,\ldots,N_k}\\
\end{array}
$$
where $X$ is assumed to range over a countably infinite set of
\emph{fixpoint variables} (letters $Y$, $Z$ will also be used to range
over fixpoint variables that may also be thought of as
meta-variables), and where both $n,k\geq0$ are arbitrary. Below, when
we refer to \emph{finitary terms} we have in mind the terms of
$\olfsfix$. The fixed-point operator is called $\gfp$ (``greatest
fixed point'') to indicate that its semantics is (now) defined in
terms of infinitary syntax, but there, fixed points are unique. Hence,
the reader may just read this as ``the fixed point''.

We now give a straightforward interpretation of the formal fixed
points (built with $\gfp$) of $\olfsfix$ in terms of the coinductive
syntax of $\coolfs$ (using the $\fix$ operation on the meta-level).

\begin{definition}\label{def:interpretation}
We call \emph{environment} a function from the set of fixpoint
variables into the set of (co)terms of $\coolfs$. The interpretation
of a finitary term (relative to an environment) is a (co)term of
$\coolfs$ given via a family of functions
$\interp{-}{\xi}:\olfsfix\rightarrow\coolfs$  indexed by
environments, which is recursively defined as follows:

$$
\begin{array}{rcll}
\interp{X}\xi& = & \xi(X)\\
\interp{\lambda x^A.N}{\xi}& = & \lambda x^A.\interp N\xi\\
\interp{\gfp\,{X}.\s{i}{E_i}}\xi&= & \fix\, N.\s i{\interp {E_i}{\xi\cup[X\mapsto N]}}\\
\interp{x \tuple{N_1,\ldots,N_k}}\xi&= & x \tuple{\interp{N_1}\xi,\ldots,\interp{N_k}\xi}\\
\end{array}
$$
where the notation $\xi\cup[X\mapsto N]$ stands for the environment obtained from $\xi$ by setting $X$ to $N$.
\end{definition}

Remark that the recursive definition above has an embedded corecursive case (pertaining to the $\gfp$-operator).
Its definition is well-formed since every elimination alternative starts with a head/application variable and the occurrences of $N$ are thus guarded.

When a finitary term $N$ has no free occurrences of fixpoint variables,  all environments determine the same coterm,
and in this case  we simply write $\interpwe N$ to denote that coterm.


%% file: equivalence.tex
\subsection{Equivalence of the
representations: Horn case}\label{subsec:equivalence-Horn-case}

\begin{theorem}[Equivalence for the Horn fragment]\label{thm:Horn}
Let $\Gamma$ be a Horn context. Then, for any atom $r$, there exists
$N_r\in\olfsfix$ with no free occurrences of fixpoint variables such
that $\interpwe{N_r}=\sol\Gamma r$.
\end{theorem}
\begin{proof}

Let us assume there are $k$ atoms occurring in $\seq\Gamma r$.
We define simultaneously $k$ functions $\N p{\vect{X:q}}$ (one for
each atom $p$ occurring in $\seq\Gamma r$), parameterized by a
vector of declarations of the form $X:q$.
The vector is written
$\vect{X:q}$ and is such that no fixpoint variable and no atom
occurs twice. The simultaneous definition is by recursion on the
number of atoms of $\seq\Gamma r$ not occurring in $\vect{X:q}$, and
is as follows:
$$
\N p{\vect{X:q}}= \left\{
\begin{array}{ll}
X_i&\mbox{if $p=q_i$}\\
\gfp\,{X_p}.{\s{(y:\vect {r}\impl p)\in\Gamma}{y \fl j{\N{r_j}{\vect
{X:q},X_p:p}}}}&\mbox{otherwise}
\end{array}
\right.
$$
where vector $\vect {X:q},X_p:p$ is obtained by adding the component
$X_p:p$ to the vector $\vect{X:q}$. Observe that only fixpoint
variables among the fixpoint variables declared in the vector have
free occurrences in $\N p{\vect{X:q}}$.

By induction on the number of atoms of (the fixed sequent) $\seq\Gamma
r$ not in (the variable) $\vect {X:q}$, we prove that:
\begin{equation}\label{eq:Horn-thm}
\interp{\N p{\vect{X:q}}}{\xi}=\sol\Gamma p\;\textrm{if}\;
\xi(X_i)=\sol\Gamma{q_i},\;\textrm{for any}\;i.
\end{equation}
Case $p=q_i$, for some $i$. Then,
$$
LHS=\interp{X_i}\xi=\xi(X_i)=\sol\Gamma{q_i}=RHS.
$$
Otherwise,
$$LHS=\interp{\gfp\,{X_p}.{\s{(y:\vect {r}\impl p)\in\Gamma}{y \fl j{\N{r_j}{\vect{X:q},X_p:p}}}}}\xi=N^\infty$$
where $N^\infty$ is given as the unique solution of the following equation:
\begin{align}
\label{eq:N-infty} N^\infty&=&\s {(y:\vect {r}\impl
p)\in\Gamma}{y\fl j {\interp
{\N{r_j}{\vect{X:q},X_p:p}}{\xi\cup[X_p\mapsto N^\infty]}}}
\end{align}
Now observe that, by I.H., the following equations (\ref{eq:sol})
and (\ref{eq:sol-two}) are equivalent.
\begin{align}
\label{eq:sol}
\sol\Gamma p&=& \s {(y:\vect {r}\impl p)\in\Gamma}{y\fl j {\interp {\N{r_j}{\vect{X:q},X_p:p}}{\xi\cup[X_p\mapsto\sol\Gamma p]}}}\\
\label{eq:sol-two} \sol\Gamma p&=& \s {(y:\vect {r}\impl
p)\in\Gamma}{y\fl j {\sol\Gamma {r_j}}}
\end{align}
By definition of $\sol\Gamma p$, (\ref{eq:sol-two}) holds; hence --
because of (\ref{eq:sol}) --
$\sol\Gamma p$ is the solution $N^\infty$ of (\ref{eq:N-infty}),
concluding the proof that $LHS=RHS$.

Finally, the theorem follows as the particular case of
(\ref{eq:Horn-thm}) where $p=r$ and the vector of fixpoint variable
declarations is empty.
\end{proof}

\subsection{Equivalence of the
representations: full implicational 
case}\label{subsec:equivalence-full-prop-case}

The main difference with exhaustive proof search in the case of Horn
formulas is that the backwards application of \RIntro\ brings new
variables into the context that may have the same type as an already
existing declaration, and so, for the purpose of
proof search, they should be treated the same way.

We illustrate this phenomenon with the following definition and lemma
and then generalize it to the form that will be needed for the main
theorem (Theorem~\ref{thm:FullProp}).

\begin{definition}
For $N$ and $E$ in $\coolfs$, we define $[\ns{x_1}{x_n}/y]N$
and $[\ns{x_1}{x_n}/y]E$ by simultaneous corecursion
as follows:
$$
\begin{array}{lcll}
{[}\ns{x_1}{x_n}/y](\lb x^A.N)&=&\lb x^A.[\ns{x_1}{x_n}/y]N\\
{[}\ns{x_1}{x_n}/y]\s i{E_i}&=&\s i{[\ns{x_1}{x_n}/y]E_i}\\
{[}\ns{x_1}{x_n}/y]\big(z\fl i{N_i}\big)&=&z\fl i{[\ns{x_1}{x_n}/y]N_i}&\textrm{if $z \neq y$}\\
{[}\ns{x_1}{x_n}/y]\big(y\fl i{N_i}\big)&=&\s{1\leq j\leq n}{x_j}\fl i{[\ns{x_1}{x_n}/y]N_i}\\
\end{array}
$$
\end{definition}

\begin{lemma}[Co-contraction: invertibility of contraction]\label{lem:cleavage}
If $x_1,x_2,y\notin\Gamma$, then
$$\sol{\Gamma,x_1:A,x_2:A}C=[\bs{x_1}{x_2}/y]\sol{\Gamma,y:A}C\enspace.$$
\end{lemma}
\begin{proof} The proof is omitted since Lemma \ref{lem:cleavage-2} below is essentially a generalization of this result. \end{proof}

We now capture when a context $\Gamma'$ is an inessential extension of context $\Gamma$:
\begin{definition}\label{def:leq}
\begin{enumerate}
\item $|\Gamma|=\{A:\exists x \textrm{ s.t.} (x:A)\in\Gamma\}$.
\item $\Gamma\leq\Gamma'$ if $\Gamma\subseteq\Gamma'$ and
$|\Gamma|=|\Gamma'|$.
\item $(\seq\Gamma p)\leq(\seq{\Gamma'}{p'})$ if $\Gamma\leq\Gamma'$ and
$p=p'$.
\end{enumerate}
\end{definition}

Let $\sigma$ range over sequents of the form $\seq{\Gamma}p$. Thus,
the last definition clause defines in general when
$\sigma\leq\sigma'$.

\begin{definition}
\begin{enumerate}
\item Let $\Gamma\leq\Gamma'$. For $N$ and $E$ in $\coolfs$, we define $[\Gamma'/\Gamma]N$
and $[\Gamma'/\Gamma]E$ by simultaneous corecursion
as follows:
$$
\begin{array}{lcll}
{[}\Gamma'/\Gamma](\lb x^A.N)&=&\lb x^A.[\Gamma',(x:A)/\Gamma,(x:A)]N\\
{[}\Gamma'/\Gamma]\s i{E_i}&=&\s i{[\Gamma'/\Gamma]E_i}\\
{[}\Gamma'/\Gamma]\big(z\fl i{N_i}\big)&=&z\fl i{[\Gamma'/\Gamma]N_i}&\textrm{if $z\notin dom(\Gamma)$}\\
{[}\Gamma'/\Gamma]\big(z\fl
i{N_i}\big)&=&\s{(w:\Gamma(z))\in\Gamma'}{w}\fl
i{[\Gamma'/\Gamma]N_i}&\textrm{if $z\in dom(\Gamma)$}
\end{array}
$$
\item Let $\sigma\leq\sigma'$. $[\sigma'/\sigma]N=[\Gamma'/\Gamma]N$ where $\sigma=(\seq{\Gamma}p)$ and $\sigma'=(\seq{\Gamma'}p)$.
Similarly for $[\sigma'/\sigma]E$.
\end{enumerate}
\end{definition}

\begin{lemma}[Co-contraction]\label{lem:cleavage-2}
If $\Gamma\leq\Gamma'$ then
$\sol{\Gamma'}C=[\Gamma'/\Gamma](\sol{\Gamma}C)$.
\end{lemma}
\begin{proof} Let
$R:=\{(\sol{\Gamma'}C,[\Gamma'/\Gamma](\sol{\Gamma}C))\mid\Gamma\leq\Gamma',C\textrm{ arbitrary}\}$.
\noindent We prove that $R$ is backward closed relative to the canonical equivalence $=$ generated by the coinductive definition of terms of $\coolfs$ (but see the comments following the proof), whence $R\subseteq=$.
\begin{equation}\label{eq:lem-cleavage2-fst}
\sol{\Gamma'}C=\lb z_1^{A_1}\cdots z_n^{A_n}.\s{(z:\vec B\impl
p)\in\Delta'}{z\fl j{\sol{\Delta'}{B_j}}}
\end{equation}
\noindent and
\begin{equation}\label{eq:lem-cleavage2-snd}
[\Gamma'/\Gamma](\sol{\Gamma}C)=\lb z_1^{A_1}\cdots
z_n^{A_n}.
\s{(y:\vec B\impl p)\in\Delta}{\s{\setbox0=\hbox{\scriptsize{\phantom{$\vec B$}}}\copy0\kern-\wd0(w:\Delta(y))\in\Delta'}{w\fl j{[\Delta'/\Delta]\sol{\Delta}{B_j}}}}
\end{equation}
\noindent where $\Delta:=\Gamma\cup\{z_1:A_1,\cdots,z_n:A_n\}$ and
$\Delta':=\Gamma'\cup\{z_1:A_1,\cdots,z_n:A_n\}$.

From $\Gamma\leq\Gamma'$ we get $\Delta\leq\Delta'$, hence
$$
(\sol{\Delta'}{B_j},[\Delta'/\Delta]\sol{\Delta}{B_j})\in R\enspace.
$$
\noindent To conclude the proof, it suffices to show that (i) each
head-variable $z$ that is a ``capability'' of the summation in
(\ref{eq:lem-cleavage2-fst}) is matched by a head-variable $w$ that
is a ``capability'' of the summation in
(\ref{eq:lem-cleavage2-snd}); and (ii) vice-versa.

(i) Let $z\in dom(\Delta')$. We have to exhibit $y\in dom(\Delta)$
such that $(z:\Delta(y))\in\Delta'$. First case: $z\in
dom(\Delta)$. By $\Delta\leq\Delta'$, $(z:\Delta(z))\in\Delta'$. So
we may take $y=z$. Second and last case:
$z\in\Gamma'\backslash\Gamma$. By $\Gamma\leq\Gamma'$, there is
$y\in\Gamma$ such that $(z:\Gamma(y))\in\Gamma'$. But then
$(z:\Delta(y))\in\Delta'$.

(ii) We have to show that, for all $y\in dom(\Delta)$, and all
$(w:\Delta(y))\in\Delta'$, $w\in dom(\Delta')$. But this is
immediate. \end{proof}

Notice that we cannot expect that the summands appear in the same
order in (\ref{eq:lem-cleavage2-fst}) and
(\ref{eq:lem-cleavage2-snd}). Therefore, we have to be more careful
with the notion of equality of B\"ohm forests. It is not just
bisimilarity, but we assume that the sums of elimination alternatives
are treated as if they were sets of alternatives, i.\,e., we further
assume that $+$ is symmetric and idempotent. It has been shown by Picard and the second author
\cite{PicardMatthesCMCS12} that bisimulation up to permutations in
unbounded lists of children can be managed in a coinductive type even
with the interactive proof assistant Coq. In analogy, this coarser
notion of equality (even abstracting away from the number of
occurrences of an alternative) should not present a major obstacle for
a fully formal presentation.

In the rest of the paper -- in particular in Theorem~\ref{thm:FullProp}
-- we assume that sums of alternatives are treated as if they were
sets.

\begin{example}[Example~\ref{ex:dn-Pierce} continued]
Thanks to the preceding lemma, $N_9$ is obtained by
co-contraction from $N_5$:
$$N_9=[x:\cdot,y:(p\impl q)\impl p,z:p,y_1:(p\impl q)\impl p,z_1:p\,/\,x:\cdot,y:(p\impl q)\impl p,z:p]N_5\enspace,$$
where the type of $x$ has been omitted. Hence, $N_6$, $N_7$, $N_8$ and $N_9$ can be eliminated, and $N_5$ can be expressed as the (meta-level) fixed point:
$$N_5=\fix\, N.x\tuple{\lambda y_1^{(p\impl q)\impl p}.y\tuple{\lambda z_1^p.[x,y,z,y_1,z_1/x,y,z]N}+z+y_1\tuple{\lambda z_1^p.[x,y,z,y_1,z_1/x,y,z]N}}\enspace,$$
now missing out all types in the context substitution.
Finally, we obtain the closed B\"ohm forest
$$\sol{}{A}=\lambda x^{(((p\impl q)\impl p)\impl p)\impl q}.x\tuple{\lambda y^{(p\impl q)\impl p}.y\tuple{\lambda z^{p}. N_5}}$$
\end{example}

The question is now how to give a finitary meaning to terms like
$N_5$ in the example above, which are defined by fixed points over
variables subject to context substitution. We might expect to use
the equation defining $N_5$ to obtain a finitary representation in
$\olfsfix$, provided context substitution is defined on this system.
But how to do that? Applying say $[x,y,z,y_1,z_1/x,y,z]$ to a plain
fixed-point variable cannot make much sense.

The desired finitary representation in the full implicational
case is obtained by adjusting the terms of $\olfsfix$ used in the
Horn case as follows:
$$
\begin{array}{lcrcl}
\textrm{(terms)} &  & N & ::= & (\cdots)|\, \gfp\,{X^{\sigma}}.\ns{E_1}{E_n}\,|\,X^{\sigma}\\
\end{array}
$$
\noindent Hence fixpoint variables are ``typed'' with
\emph{sequents} $\sigma$.

Different free occurrences of the same $X$ may be "typed" with
different $\sigma$'s, as long as a lower bound of these $\sigma$'s
can be found w.r.t. $\leq$ (Definition~\ref{def:leq}).

Relatively to Definition \ref{def:interpretation}, an environment
$\xi$ now assigns (co)terms $N$ of $\coolfs$ to ``typed'' fixpoint
variables $X^{\sigma}$, provided $X$ does not occur with two
different ``types'' in the domain of $\xi$, for all $X$; we also
change the following clauses:
$$
\begin{array}{rcll}
\interp{X^{\sigma'}}\xi& = & [\sigma'/\sigma]\xi(X^{\sigma})&\textrm{if $\sigma\leq\sigma'$}\\
\interp{\gfp\,{X^{\sigma}}.\s{i}{E_i}}\xi&= & \fix\, N.\s i{\interp {E_i}{\xi\cup[X^{\sigma}\mapsto N]}}\\
\end{array}
$$
We will have to assign some default value to $X^{\sigma'}$ in case there is no such $\sigma$, but this will not play a role in the main result below.

Map $\N p{\vect{X:q}}$ used in the proof of Theorem \ref{thm:Horn}
is replaced by the following:

\begin{definition}\label{def:map-N} Let $\Xi:=\vect{X:\seq{\Theta}q}$ be a vector
of $m\geq 0$ declarations $(X_i:\seq{\Theta_i}{q_i})$ where no
fixpoint variable and no sequent occurs twice. $\N{\seq\Gamma{\vec
A\impl p}}{\Xi}$ is defined as follows:

If, for some $1\leq i\leq m$, $p=q_i$ and $\Theta_i\subseteq\Gamma$
and $|\Theta_i|=|\Delta|$, then
$$
\N{\seq\Gamma{\vec A\impl p}}{\Xi}=\lb z_1^{A_1}\cdots
z_n^{A_n}.X_i^{\sigma}
$$
otherwise,
$$
\N{\seq\Gamma{\vec A\impl p}}{\Xi}=\lb z_1^{A_1}\cdots
z_n^{A_n}.\gfp\,{Y^{\sigma}}.{\s{(y:\vec {B}\impl p)\in\Delta}{y \fl
j{\N{\seq{\Delta}{B_j}}{\Xi,Y:\sigma}}}}
$$
where, in both cases, $\Delta:=\Gamma\cup\{z_1:A_1,\cdots,z_n:A_n\}$
and $\sigma:=\seq{\Delta}p$.
\end{definition}

The definition of $\N p{\vect{X:q}}$ in the proof of Theorem
\ref{thm:Horn} was by recursion on a certain number of atoms. The
following lemma spells out the measure that is recursively decreasing in the
definition of $\N{\seq\Gamma{C}}{\Xi}$.

\begin{lemma}\label{lem:termination}
For all $\seq\Gamma C$, $\N{\seq\Gamma C}{\cdot}$ is well-defined,
where $\cdot$ denotes the empty vector.
\end{lemma}
\begin{proof} Let us call \emph{recursive call} a ``reduction''
\begin{equation}\label{eq:recursive-call}
\N{\seq\Gamma{\vec A\impl p}}{\vect{X:\seq{\Theta}q}}\leadsto
\N{\seq{\Delta}{B_j}}{\vect {X:\seq{\Theta}q},Y:\sigma}
\end{equation}
\noindent where the if-guard in Def.~\ref{def:map-N} fails; $\Delta$ and
$\sigma$ are defined as in the same definition; and, for some $y$,
$(y:\vec {B}\impl p)\in\Delta$. We want to prove that every sequence
of recursive calls from $\N{\seq\Gamma C}{\cdot}$ is finite.

First we introduce some definitions. ${\cal
A}^{sub}:=\{B\mid\textrm{there is $A\in{\cal A}$ such that $B$
is subformula of $A$}\}$, for $\cal A$ a finite set of formulas. We
say $\cal A$ is \emph{subformula-closed} if ${\cal A}^{sub}={\cal
A}$. A \emph{stripped sequent} is a pair $({\cal B},p)$, where $\cal
B$ is a finite set of formulas. If $\sigma=\seq{\Gamma}p$, then
$|\sigma|$ denotes the stripped sequent $(|\Gamma|,p)$. We say
$({\cal B},p)$ \emph{is over} $\cal A$ if ${\cal B}\subseteq{\cal
A}$ and $p\in{\cal A}$. There are $size({\cal A}):=a\cdot2^k$ stripped
sequents over $\cal A$, if $a$ (resp. $k$) is the number of atoms
(resp. formulas) in $\cal A$.

Let $\cal A$ be subformula-closed. We say $\seq\Gamma{C}$ and
$\Xi:=\vect{X:\seq{\Theta}q}$ satisfy the \emph{$\cal A$-invariant}
if:
\begin{itemize}
\item[(i)] $|\Gamma|\cup\{C\}\subseteq{\cal A}$;
\item[(ii)]
$\Theta_1\subseteq\Theta_2\subseteq\cdots\subseteq\Theta_m=\Gamma$ (if $m=0$ then this is meant to be vacuously true);
\item[(iii)] For $1\leq j\leq m$, $q_j\in|\Gamma|^{sub}$,
\end{itemize}
where $m\geq 0$ is the length of vector $\Xi$ (if $m=0$, also item
(iii) is vacuously true). In particular, $|\sigma|$ is over $\cal
A$, for all $\sigma\in\Xi$. We prove that, if $\seq\Gamma C$ and
$\Xi$ satisfy the $\cal A$-invariant for some $\cal A$, then every
sequence of recursive calls from $\N{\seq\Gamma C}{\Xi}$ is finite.
The proof is by induction on $size({\cal A})-size(\Xi)$, where
$size(\Xi)$ is the number of elements of $|\Xi|$ and
$|\Xi|:=\{|\sigma|:\sigma\in\Xi\}$.

Let $C=\vec A\impl p$. We analyze an arbitrary recursive call
(\ref{eq:recursive-call}) and prove that every sequence of recursive
calls from $\N{\seq{\Delta}{B_j}}{\Xi,Y:\sigma}$ is finite. This is
achieved by proving:
\begin{itemize}
\item[(I)] $\seq{\Delta}{B_j}$ and $\Xi,Y:\sigma$ satisfy the $\cal
A$-invariant;
\item[(II)] $size(\Xi,Y:\sigma)>size(\Xi)$.
\end{itemize}

Proof of (I). By assumption, (i), (ii), and (iii) above hold. We
want to prove:
\begin{itemize}
\item[(i')] $|\Delta|\cup\{B_j\}\subseteq{\cal A}$;
\item[(ii')]
$\Theta_1\subseteq\Theta_2\subseteq\cdots\subseteq\Theta_m\subseteq\Delta=\Delta$;
\item[(iii')] For $1\leq j\leq m+1$, $q_j\in|\Delta|^{sub}$.
\end{itemize}

Proof of (i').
$|\Delta|=|\Gamma|\cup\{A_1,\cdots,A_n\}\subseteq{\cal A}$ by (i)
and $\cal A$ subformula-closed. $B_j$ is a subformula of $\vec
B\impl p$ and $\vec B\impl p\in|\Delta|$ because $(y:\vec B\impl
p)\in\Delta$, for some $y$.

Proof of (ii'). Immediate by (ii) and $\Gamma\subseteq\Delta$.

Proof of (iii'). For $1\leq j\leq m$,
$q_j\in|\Gamma|^{sub}\subseteq|\Delta|^{sub}$, by (iii) and
$\Gamma\subseteq\Delta$. On the other hand,
$q_{j+1}=p\in|\Delta|^{sub}$ because $(y:\vec B\impl p)\in\Delta$,
for some $y$.

Proof of (II). Given that the if-guard of Def.~\ref{def:map-N}
fails, and that $\Theta_i\subseteq\Gamma$ due to (ii), we conclude:
for all $1\leq i\leq m$, $p\neq q_i$ or $|\Theta_i|\neq|\Delta|$.
But this means that $|\seq{\Delta}p|\notin|\Xi|$, hence
$size(\Xi,Y:\sigma)>size(\Xi)$.

Now, by I.H., every sequence of recursive calls from
$\N{\seq{\Delta}{B_j}}{\Xi,Y:\sigma}$ is finite. This concludes the
proof by induction.

Finally let ${\cal A}=(|\Gamma|\cup\{C\})^{sub}$ and observe that
$\seq\Gamma C$ and $\Xi=\cdot$ satisfy the $\cal A$-invariant.
\end{proof}

\begin{theorem}[Equivalence]\label{thm:FullProp}
For any $\Gamma$ and $C$, there exists $N_{\seq\Gamma C}\in\olfsfix$
with no free occurrences of fixpoint variables such that
$\interpwe{N_{\seq\Gamma C}}=\sol\Gamma C$.
\end{theorem}
\begin{proof} We prove: if, for all $i$, $\xi(X_i^{\seq{\Theta_i}{q_i}})=\sol{\Theta_i}{q_i}$, then
\begin{equation}\label{eq:full-thm}
\interp{\N{\seq\Gamma{\vec A\impl p}}{\Xi}}{\xi}=\sol\Gamma{\vec
A\impl p}\enspace,
\end{equation}
\noindent where $\Xi:=\vect{X:\seq{\Theta}q}$. In this proof we
re-use the concepts introduced in the proof of Lemma
\ref{lem:termination}. Let ${\cal{A}}:=(|\Gamma|\cup\{\vec A\impl
p\})^{sub}$. The proof is by induction on $size({\cal
A})-size(\Xi)$.

Case $p=q_i$ and $\Theta_i'\subseteq\Gamma$ and
$|\Theta_i'|=|\Delta|$, for some $1\leq i\leq m$, with $m$ the
length of $\Xi$.
Then,
$$
\begin{array}{rcll}
LHS&=&\lb z_1^{A_1}\cdots
z_n^{A_n}.\interp{X_i^{\seq{\Delta}{q_i}}}\xi&\textrm{(by definition)}\\
&=&\lb z_1^{A_1}\cdots
z_n^{A_n}.[\seq{\Delta}{q_i}/\seq{\Theta_i}{q_i}]\xi(X_i^{\seq{\Theta_i}{q_i}})&\textrm{(by definition and (*) below)}\\
&=&\lb z_1^{A_1}\cdots
z_n^{A_n}.[\seq{\Delta}{q_i}/\seq{\Theta_i}{q_i}]\sol{\Theta_i}{q_i}&\textrm{(by assumption)}\\
&=&\lb z_1^{A_1}\cdots
z_n^{A_n}.\sol{\Delta}{q_i}&\textrm{(by Lemma \ref{lem:cleavage-2} and (*))}\\
&=&RHS&\textrm{(by definition)}
\end{array}
$$
\noindent where $\Delta:=\Gamma\cup\{z_1:A_1,\cdots,z_n:A_n\}$,
which implies $(\seq{\Theta_i}{q_i})\leq(\seq{\Delta}{q_i})$. The latter
fact is the justification (*) used above.

The inductive case is an easy extension of the inductive case in
Theorem \ref{thm:Horn}. Suppose the case above holds for no $1\leq
i\leq m$. Then $LHS=\lb z_1^{A_1}\cdots
z_n^{A_n}.N^\infty$, where $N^\infty$ is the unique solution of the following equation
\begin{eqnarray}
\label{eq:N-infty-full-case} N^\infty&=&\s {(y:\vect {B}\impl
p)\in\Delta}{y\fl j {\interp
{\N{\seq{\Delta}{B_j}}{\Xi,Y:\sigma}}{\xi\cup[Y^{\sigma}\mapsto
N^\infty]}}}
\end{eqnarray}
\noindent and, again,
$\Delta:=\Gamma\cup\{z_1:A_1,\cdots,z_n:A_n\}$. Now observe that, by
I.H., the following equations (\ref{eq:sol-full-case}) and
(\ref{eq:sol-two-full-case}) are equivalent.
\begin{eqnarray}
\label{eq:sol-full-case} \sol\Delta{p}&=& \s {(y:\vect {B}\impl p)\in\Delta}{y\fl j {\interp {\N{\seq{\Delta}{B_j}}{\Xi,Y:\sigma}}{\xi\cup[Y^{\sigma}\mapsto\sol\Delta p]}}}\\
\label{eq:sol-two-full-case} \sol\Delta{p}&=& \s {(y:\vect {B}\impl
p)\in\Delta}{y\fl j {\sol\Delta {B_j}}}
\end{eqnarray}
By definition of $\sol\Delta{p}$, (\ref{eq:sol-two-full-case})
holds; hence - because of (\ref{eq:sol-full-case}) -
$\sol\Delta p$ is the solution $N^\infty$ of
(\ref{eq:N-infty-full-case}). Therefore $LHS=\lb z_1^{A_1}\cdots
z_n^{A_n}.\sol\Delta p$, and the latter is $RHS$ by definition of
$\sol\Gamma{\vec A\impl p}$.

Finally, the theorem follows as the particular case of
(\ref{eq:full-thm}) where $C=\vec A\impl p$ and the vector of
fixpoint variable declarations is empty. \end{proof}


%% file: conclusion.tex
\section{Conclusion}\label{sec:conclusion}

We proposed a coinductive approach to proof search, which we illustrated in the case of the cut-free system $LJT$ for intuitionistic implication (and its proof-annotated version $\ol$). As the fundamental tool, we introduced the coinductive calculus $\coolfs$, which besides the coinductive reading of $\ol$, introduces a construction for finite alternatives. The (co)terms of this calculus (also called B\"{o}hm forests) are used to represent the solution space of proof search for $LJT$-sequents, and this is achieved by means of a corecursive function, whose definition arises naturally by taking a 
reductive view of the inference rules and by using the finite alternatives construction to account for multiple alternatives in deriving a given sequent.

We offered also a finitary representation of proof search in $LJT$, based on  the inductive calculus $\olfsfix$ with finite alternatives and a fixed point construction, and showed equivalence of the representations. The equivalence results turned out to be an easy task in the case of the Horn fragment, but demanded for
co-contraction of contexts (contraction bottom-up) in the case of full implication.

With Pym and Ritter \cite{pym2004reductive} we share the general goal of setting a framework for studying proof search, and the reductive view of inference rules, by which each inference rule is seen as a reduction operator (from a putative conclusion to a collection of sufficient premises), and reduction (the process of repeatedly applying reduction operators) may fail to yield a (finite) proof. However, the methods are very different. Instead of using a coinductive approach, Pym and Ritter 
introduce the $\lambda\mu\nu\epsilon$-calculus for classical sequent calculus as the means for representing derivations and for studying intuitionistic proof search (a task that is carried out both in the context of the sequent calculus LJ and of intuitionistic resolution).

In the context of logic programming with classical first-order Horn clauses, and building on their previous work \cite{DBLP:conf/calco/KomendantskayaP11,DBLP:conf/amast/KomendantskayaMP10}, Komendantskaya and Power \cite{DBLP:conf/csl/KomendantskayaP11} establish
a coalgebraic semantics uniform for both finite and infinite SLD-resolutions.
In particular, a notion of coinductive (and-or) derivation tree of an atomic goal w.\,r.\,t.~a (fixed) program is introduced.
Soundness and completeness results of SLD-resolution relative to coinductive derivation trees and to the coalgebraic semantics are also proved.
Logic programming 
is viewed as search for  uniform proofs in sequent calculus by Miller \emph{et al.} \cite{DBLP:journals/apal/MillerNPS91}. 
For intuitionistic implication, uniform proofs correspond to the class of ($\eta$-)expanded normal natural deductions  (see Dyckoff and Pinto \cite{DPPSTTL94}), hence to the typed $\ol$-terms we considered in this paper (recall the restriction to atoms in rule $Der$  of Fig. \ref{fig:lambda-bar} for typing application). Under this view, our work relates to Komendantskaya and Power \cite{DBLP:conf/csl/KomendantskayaP11}, as both works adopt a coinductive approach
in the context of proof search. However, the two approaches are different in methods and in goals.  As the basis of
the coinductive representation of the search space, instead of and-or infinite trees, we follow the Curry-Howard view of proofs as terms,
and propose the use of a typed calculus of coinductive lambda-terms.
Whereas Komendantskaya and Power \cite{DBLP:conf/csl/KomendantskayaP11} are already capable of addressing first-order quantification, we only consider intuitionistic implication.
Still, as we consider full intuitionistic implication, our study is not contained in classical Horn logic.
The fact that we need to treat negative occurrences of implication,
raises on the logic programming side the need for dealing with programs to which clauses can be added dynamically.

As a priority for future work, we plan to develop notions of normalisation for the calculi $\coolfs$ and $\olfsfix$ in connection with aspects of proof search like pruning search spaces and reading off (finite) proofs.

In order to test for the generality of our approach, we intend to extend it to treat the first-order case. Staying within intuitionistic implication,
but changing the proofs searched for, another case study we intend to investigate is Dyckhoff's contraction-free system \cite{DBLP:journals/jsyml/Dyckhoff92}.


\paragraph{Acknowledgments}

We thank our anonymous referees for their helpful comments. Jos\'{e}
Esp\'{\i}rito Santo and Lu\'{\i}s Pinto have been financed by FEDER
funds through ``Programa Operacional Factores de Competitividade --
COMPETE'' and by Portuguese funds through FCT -- ``Funda\c{c}\~{a}o
para a Ci\^{e}ncia e a Tecnologia'', within the project
PEst-C/MAT/UI0013/2011.  Ralph Matthes thanks the Centro de
Matem\'{a}tica of Universidade do Minho for funding research visits to
Jos\'{e} Esp\'{\i}rito Santo and Lu\'{\i}s Pinto to start this
research (2011/2012). Subsequently, he has been funded by the \emph{Climt} project
(ANR-11-BS02-016 of the French Agence Nationale de la Recherche).

